\documentclass[10pt]{article}

\newcommand{\ket}[1]{|#1\rangle}

\usepackage[margin=3cm, a4paper]{geometry}
\usepackage{amsmath,amsthm,amsfonts,amssymb}
\usepackage{fontenc}
\usepackage[all]{xy}
\usepackage{graphicx,subfigure,multirow}
\usepackage[figurename=Fig.]{caption}
\usepackage{tikz}
\usetikzlibrary{shapes, arrows, shadows}
\usepackage[latin1]{inputenc}
\usepackage[scaled]{helvet}
\usepackage[toc, page]{appendix}

\newtheorem{theorem}{Theorem}[subsection]
\newtheorem{definition}[theorem]{Definition}

\newtheorem{lemma}[theorem]{Lemma}

\newtheorem{Corollary}[theorem]{Corollary}

\numberwithin{equation}{subsection}

\begin{document} 

\begingroup  
\centering
{\Large\textbf{Bounds on the power of proofs and advice in general physical theories} \\[1.5em]
  \normalsize Ciar\'{a}n M. Lee\footnote{
Electronic address: ciaran.lee@cs.ox.ac.uk}  and Matty J. Hoban}\\[1em]
{\it  University of Oxford, Department of Computer Science, Wolfson Building, Parks Road, Oxford OX1 3QD, UK.}\\[1em]
\endgroup

\begin{abstract} 
Quantum theory presents us with the tools for computational and communication advantages over classical theory. One approach to uncovering the source of these advantages is to determine how computation and communication power vary as quantum theory is replaced by other operationally-defined theories from a broad framework of such theories. Such investigations may reveal some of the key physical features required for powerful computation and communication. In this paper we investigate how simple physical principles bound the power of two different computational paradigms which combine computation and communication in a non-trivial fashion: computation with advice and interactive proof systems. We show that the existence of non-trivial dynamics in a theory implies a bound on the power of computation with advice. Moreover, we provide an explicit example of a theory with no non-trivial dynamics in which the power of computation with advice is unbounded. Finally we show that the power of simple interactive proof systems in theories where local measurements suffice for tomography is non-trivially bounded. This result provides a proof that $\bold{QMA}$ is contained in $\bold{PP}$ which does not make use of any uniquely quantum structure -- such as the fact that observables correspond to self-adjoint operators -- and thus may be of independent interest.

\end{abstract} 




\section{Introduction}    

\subsection{Motivation}  

Since the mid $1980$s there has been growing evidence that quantum theory offers dramatic advantages in both computation and communication problems \cite{arkhipov,shor,Nielsen-Chuang, Raz, Dam}. In particular, the existence of an efficient quantum algorithm for factoring \cite{shor} and of a communication problem for which quantum theory requires exponentially less communication to solve \cite{Raz} has challenged classical conceptions of what problems are efficiently solvable in our physical world. 

Much recent work has been concerned with uncovering the source of this quantum advantage \cite{Magic,Vidal,Hoban-2011,Hoban-Browne,Raussendorf,discord,Interference-speed-up,anti-discord,entanglement, BellAdvantage}. One approach to this problem is to view quantum theory in the context of a framework general enough to accommodate essentially any operationally-defined theory \cite{LB-2014, LS-2015}. While most of these theories may not correspond to descriptions of our physical world, they nevertheless make good operational sense and allow one to systematically assess how computation and communication power depend on the underlying physical theory. Determining how computation and communication power vary as quantum theory is replaced by other operationally-defined theories may reveal some of the key physical features required for powerful computation and communication. 
  
More generally, within this framework, one can identify physical principles that theories may or may not satisfy, such as causality (no signalling from future to past), or tomographic locality (local measurements suffice for tomography of joint states). It's recently been shown\footnote{Other investigations linking physical principles to computation can be found in \cite{LB-2014, BBHL-2015, LS-2015, LS1-2015}.} \cite{LB-2014} that for any theory satisfying tomographic locality, whether or not causality is satisfied, computational problems that can be solved efficiently are contained in the classical complexity class $\bold{AWPP}$ -- a fact first proved in the quantum case by Fortnow and Rogers \cite{AWPP}. 

In this paper we investigate how simple physical principles bound the power of two different computational paradigms which combine both computation and communication in a non-trivial fashion: computation with advice and interactive proof systems. These are standard tools in computational complexity and one can view our work as methodically exploring the impact of general physical theories upon these tools (further expanding upon the work in \cite{LB-2014}).

\subsection{Overview of the results}  

Computation with advice considers the situation where an efficient computer is supplemented with extra information, or \emph{advice}, which, in classical computation, takes the form of a bit string and, in quantum computation, takes the form of a quantum state. The usefulness of this computational paradigm is that no so-called uniformity constraints are placed on the string or state embodying the advice -- as is usually the case when one considers efficient computation -- and so one can attempt to encode solutions to hard problems in the advice. Aaronson was among\footnote{Quantum computation with advice was first defined and studied in \cite{NY03}} the first to study and set bounds on the power of quantum computation with (quantum) advice \cite{Aaronson-advice}. His primary motivation was a desire to investigate the question ``How many classical bits can `really' be encoded into $n$ qubits?'' from a complexity theoretic point of view.

Aaronson noted \cite{Aaronson-advice} that quantum advice is quite closely related to quantum one-way communication\footnote{Quantum one-way communication can be described as follows: Alice has an $n$-bit string $x$, Bob has an $m$-bit string $y$, and together they wish to evaluate $f(x,y)$ where $f:\{0,1\}^n \times \{0,1\}^m \rightarrow \{0,1\}$ is a Boolean function. After examining her input $x=x_1\dots x_n$, Alice can send a single quantum message $\rho_x$ to Bob, whereupon Bob, after examining his input $y=y_1\dots y_m$, can choose some basis in which to measure $\rho_x$. He must then output a claimed value for $f(x,y)$. We are interested in how long Alice's message needs to be, for Bob to succeed with high probability on any $x,y$ pair.}, since one can think of an advice state as a one-way message sent to an algorithm by a benevolent ``advisor''. The class of decision problems which can be efficiently solved on a quantum computer with access to a quantum advice state is denoted $\bold{BQP/qpoly}$, and Aaronson showed \cite{Aaronson-advice} that $\bold{BQP/qpoly}\subseteq\bold{PP/poly}$. Based on the relation between quantum advice and quantum one-way communication, the size of the class $\bold{BQP/qpoly}$ can, in some sense, be thought of as a measure of prowess in communication tasks, or, intuitively speaking, as a measure of how much `useful' information can be stored in a quantum state.  

If the computational power of a general theory can be considered a measure of the richness of its dynamics, then the increase in computational power when supplemented with advice can be thought of -- \textit{\`{a} la} Aaronson in the previous discussion -- as a measure of the information that can be stored in its states. In section~\ref{general-def} we provide rigorous definitions of the class of decision problems that can be solved in a specific operational theory when provided with a trusted advice state from that theory -- which we call $\bold{BGP/gpoly}$ for a particular theory $\bold{G}$. We show that in the theory colloquially known as ``Boxworld'', which has the strongest correlations consistent with the no-signalling principle and was first discussed by Popescu and Rohrlich in \cite{PRbox, Popescu-review}, the class $\bold{BGP/gpoly}$ contains all decision problems and so is optimally powerful. Despite this, section~\ref{dynamics} shows that theories with a certain amount of non-trivial dynamics satisfy the same upper bound on the power of computation with advice as was discussed in the previous paragraph for quantum theory. In particular, for theories $\bold{G}$ with non-trivial dynamics we show that $\bold{BGP/gpoly} \subseteq \bold{PP/poly}$. Boxword has no non-trivial reversible dynamics and it was shown by van Dam \cite{PR-van-Dam} that communication complexity tasks in Boxworld can be solved trivially. Our result shows that when a theory \emph{has} non-trivial reversible dynamics there is a limit on its prowess in certain communication tasks -- as quantified by the size of the class $\bold{BGP/gpoly}$. 
 
A key point in the above discussion is that one trusts the advice provider. That is, one trusts that the received advice contains the information the provider claims it does. In reality the provider could be malevolent and out to deceive the receiver. If one cannot trust the provider, a computer must be used to check -- or \emph{verify} -- that the provided advice is correct and this verification process requires non-trivial dynamics to implement. Thus, by learning how computational complexity changes as the amount of trust we have in the provider is varied, we enter into a regime where both prowess in communication tasks and computational power -- corresponding to the existence of non-trivial dynamics -- are simultaneously tested. 


Within theoretical computer science, untrusted advice has been formally referred to as \textit{proofs} and has a long history within computational complexity. For example, the famous class $\bold{NP}$ can be described as a proof system between an efficient, deterministic, classical computer (verifier) and an all-powerful prover where the prover gives polynomially-sized proofs to the verifier. Here the verifier wishes to check if this proof is the correct solution to a particular problem. In quantum computing, the corresponding complexity class to $\bold{NP}$ is denoted $\bold{QMA}$, for Quantum Merlin-Arthur. The question of what useful problems a quantum computer can solve when given a non-uniform quantum state as a proof from an untrusted source has led to surprising and beautiful connections between quantum computation and condensed matter physics \cite{Aharonov}. 

In section~\ref{general-def}, we give a rigorous definition of the class of problems for which a verifier with an efficient computer from a specific theory can solve when given proof states from that theory -- which we call $\bold{GMA}$ for a particular theory $\bold{G}$. We show, in section~\ref{trade-off}, that there exists a universal upper bound on $\bold{GMA}$ for all causal and tomographically local theories. In particular, we show that $\bold{GMA}\subseteq\bold{PP}$ for all $\bold{G}$ satisfying tomographic locality and causality. Note that Boxworld is an example of such a theory. Some results concerning the connection between trusted advice and proof verification in general theories are given in section \ref{optimal}.

\section{The framework}  \label{gpt}
\subsection{Operational theories} 

We work in the circuit framework for generalised probabilistic theories developed by Hardy in \cite{Hardy-2011} and Chiribella, D'Ariano and Perinotti in \cite{Pavia1,Pavia2}. The presentation here is most similar to that of Chiribella \textit{et al}. We now provide a brief review of this framework, see \cite{LB-2014} for a more in-depth review and an extended discussion of computation in general theories.  

A theory within this framework specifies a set of laboratory devices that can be connected together in different ways to form experiments and assigns probabilities to different experimental outcomes. A laboratory device comes equipped with input ports, output ports, and a classical pointer. When a device is used in an experiment, the pointer comes to rest in one of a number of positions, indicating some outcome has occurred. One can intuitively think of \emph{physical systems} as passing between the input and output ports of the laboratory devices. and these physical systems come in different types, denoted by labels $A,B,C,\dots$. In an experiment these devices can be composed both sequentially and in parallel, and when composed sequentially, types must match: the output system of the first device must be of the same type as the corresponding input system of the second.
 
In a general theory, one can depict the connections of devices in some experimental set-up by closed circuits. A fundamental requirement on any physical theory is that it should be able to give probabilistic predictions about the occurrence of possible outcomes. It is thus demanded that, in this framework, closed circuits define probability distributions. Given this structure, one then says that two physical devices are equivalent (from the point of view of the theory) if replacing one by the other in any closed circuit does not change the probabilities. The set of equivalence classes of devices with no input ports are referred to as \emph{states}, devices with no output ports as \emph{effects} and devices with both input and output ports as \emph{transformations}. 

The `Dirac-like' notation $|s_{r})_A$ is used to represent a state of system type $A$, where $r$ is the outcome of the classical pointer, and $_A(e_{r}|$ to represent an effect on system type $A$, so that if the effect $_A(e_{r_2}|$ is applied to the state $|s_{r_1})_A$, the probability to obtain outcome $r_1$ on the physical device representing the state and outcome $r_2$ on the physical device representing the effect is
\[
_A(e_{r_2}|s_{r_1})_A := P(r_1,r_2).
\]
The fact that closed circuits correspond to probabilities can be leveraged to show that the set of states, effects and transformations each give rise to a vector space and that the transformations and effects act linearly on the vector space of states. We assume in this work that all vector spaces are finite dimensional.

We can now formally define some examples of physical principles. We will first discuss the principles of \emph{causality} and \emph{tomographic locality} which were briefly mentioned in the introduction section. 
\begin{definition}[Causality \cite{Pavia1}]
A theory is said to be \emph{causal} if the marginal probability of a preparing a state is independent of the choice of which measurement follows the preparation. 
\end{definition}
More formally, if $\{|s_i)\}_{i\in{X}}$ are the states corresponding to the preparation, consider the probability of outcome $i$, given that a subsequent measurement $\mathcal{M}$ corresponds to a set of effects $\{ {(e_j|\}_{j\in{Y}}}$:
\[
P(i | \mathcal{M} ):=\sum_{j\in{Y}}(e_j|s_i).
\]
The theory is causal if for any system type $A$, any preparation test with outcome $i$, and any pair of measurements, $\mathcal{M}$ and $\mathcal{N}$, with input type $A$, 
\[
P(i|\mathcal{M})=P(i|\mathcal{N}). 
\] 
One can think \footnote{Provided one thinks of circuits as having a temporal order, with tests later in the sequence occurring at a later time than tests earlier in the sequence.} of the causality principle as intuitively capturing the notion of \emph{no signalling from the future}. It was shown in \cite{Pavia1} that a theory is causal if and only if for every system type $A$, there is a unique deterministic effect $_A(u|$. In this case, a measurement, with corresponding effects $\{ {(e_j|\}_{j\in{Y}}}$, satisfies $\sum_j { (e_j |} = {(u|}$. A state $|s)$ is \emph{normalised} if and only if ${(u|s)} = 1$. It can be shown that, without loss of generality, every state in a causal theory can be taken to be normalised \cite{Pavia1}. 

\begin{definition}[Tomographic locality \cite{Pavia1, Hardy-2011, Barrett-2007}]
A theory satisfies tomographic locality if every transformation can be uniquely characterised by local process tomography. 
\end{definition}
That is, in a tomographically local theory, if two transformations, with matching input and output ports, give the same probabilities for all product state inputs and product effect measurements, then the transformations must be equivalent. Tomographic locality implies that the matrix corresponding to a composite transformation is just the vector space tensor product of the matrices of each individual transformation in the composite. 

We will now define strong symmetry, a principle, which if satisfied, guarantees the existence of a certain type of non-trivial dynamics. Before we define this principle, the following concepts must be introduced. We say the laboratory device $\{\mathcal{U}_j\}_{j\in{Y}}$, where $j$ indexes the positions of the classical pointer, is a \emph{coarse-graining} of the device $\{\mathcal{E}_i\}_{i\in{X}}$ if there is a disjoint partition $\{X_j\}_{j\in{Y}}$ of $X$ such that $\mathcal{U}_j=\sum_{i\in{X_j}}\mathcal{E}_i$. That is, coarse-graining arises when some outcomes of a laboratory device are joined together. The device $\{\mathcal{E}_i\}_{i\in{X}}$ is said to \emph{refine} the device $\{\mathcal{U}_j\}_{j\in{Y}}$. A state is \emph{pure} if it does not arise as a \emph{coarse-graining} of other states; a pure state is one for which we have maximal information. A state is \emph{mixed} if it is not pure and it is \emph{completely mixed} if any other state refines it. That is, $|c)$ is completely mixed if for any other state $|\rho)$, there exists a non-zero probability $p$ such that $p|\rho)$ refines $|c)$. States $\{|\sigma_i)\}_{i=1}^N$ are \emph{perfectly distinguishable} if there exists a measurement, corresponding to effects $\{(e_i|\}_{i=1}^N$, such that $(e_i|\sigma_j)=\delta_{ij}$ for all $i,j$. 

\begin{definition}[Strong symmetry \cite{Higher-order-reconstruction}] \label{symm}
A theory satisfies \emph{strong symmetry} if for any two $n$-tuples of pure and perfectly distinguishable states $\{|\rho_1),\dots,|\rho_n)\},$ and $\{|\sigma_1),\dots,|\sigma_n)\},$ there exists a reversible transformation $T$ such that $T|\rho_i)=|\sigma_i)$ for $i=1,\dots,n$.
\end{definition}

In section \ref{dynamics}, we will mainly be concerned with two special cases of the above principle:
\begin{enumerate}
\item \textbf{Permutability}: A general theory satisfies \emph{Permutability} if for any $n$-tuple of pure and perfectly distinguishable states and any permutation $\pi$ of this $n$-tuple
$$\{|\rho_1),\dots,|\rho_n)\} \quad\&\quad \{|\rho_{\pi(1)}),\dots,|\rho_{\pi(n)})\},$$
there exists a reversible transformation $T$ such that $T|\rho_i)=|\rho_{\pi(i)})$ for $i=1,\dots,n$.
\item \textbf{Bit-symmetry}: A theory satisfies \emph{bit-symmetry} if for any two $2$-tuples of pure and perfectly distinguishable states $\{|\rho_1),|\rho_2)\},\{|\sigma_1), |\sigma_2)\},$ there exists a reversible transformation $T$ such that $T|\rho_i)=|\sigma_i)$ for $i=1,2$.
\end{enumerate}

Permutability is the special case of definition \ref{symm} where one of the sets of pure and perfectly distinguishable states is a permutation of the other. Bit-symmetry is the $n=2$ case of definition \ref{symm}.

Note that causality, tomographic locality and strong symmetry are all logically independent: generalised probabilistic theories satisfying any subset (including the empty subset) can be defined. For example, standard quantum theory satisfies all three, quantum theory with real amplitudes satisfies causality and strong symmetry but not tomographic locality, Boxworld satisfies causality and tomographic locality but not strong symmetry and the theory constructed in \cite{cause} does not satisfy causality. 

\subsection{Efficient computation}

To define the class of efficient computation in a general theory, we must first define the notions of a uniform circuit family and an acceptance condition for an arbitrary theory. The notion of a poly-size uniform circuit family $\{C_x\}$, which is indexed by some bit string $x$, can be defined as follows:
\begin{enumerate} 
\item The number of gates in the circuit $C_x$ is bounded by a polynomial in $|x|$.
\item There is a finite \footnote{For a uniformity condition where the size of the gate grow with circuits size, see \cite{NdB}} gate set $\mathcal{G}$, such that each circuit in the family is built from elements of $\mathcal{G}$.
\item For each type of system, there is a fixed choice of basis, relative to which transformations are associated with matrices. Given the matrix ${M}$ representing (a particular outcome of) a gate in $\mathcal{G}$, a Turing machine can output a matrix $\widetilde{{M}}$ with rational entries, such that $ | ({M} - \widetilde{{M}})_{ij} | \leq \epsilon$, in time polynomial in $\log(1/\epsilon)$. 
\item There is a Turing machine that, acting on input $x=x_1x_2\dots x_n$, outputs a classical description of $C_x$ in time bounded by a polynomial in $|x|$.
\end{enumerate}

At the end of each run of the computation, each gate in the circuit has a classical outcome -- corresponding to the final position of the classical pointer -- associated with it, and the theory defines a joint probability for these outcomes. Denoting the string of observed outcomes by $z$, we define the final output of the computation to be given by a function $a(z)\in \{0,1\}$, where there must exist a Turing machine that computes $a$ in time polynomial in the length of the input $|x|$. We say the computation accepts an input string $x$ if $a(z)=0$, where $z$ is an outcome string of the circuit $C_x$. The probability that a computation accepts the input string $x$ is therefore given by
$$P_x({\mathrm{accept}})=\sum_{z|a(z)=0}P(z),$$ where the sum ranges over all possible outcome strings of the circuit $C_x$.

The class of problems that can be solved efficiently in a generalised probabilistic theory can now be defined.
\begin{definition}
For a generalised probabilistic theory $\bold{G}$, a language $\mathcal{L}$ is in the class $\bold{BGP}$ if there exists a poly-sized uniform family of circuits in $\bold{G}$, and an efficient acceptance criterion, such that
\begin{enumerate}
\item $x\in\mathcal{L}$ is accepted with probability at least $\frac{2}{3}$. 
\item $x\notin\mathcal{L}$ is accepted with probability at most $\frac{1}{3}$. \end{enumerate}  
\end{definition}  

The choice of constants $(\frac{2}{3},\frac{1}{3})$ is arbitrary as long as they are bounded away from $1/2$ by some constant \footnote{This can be further relaxed to being bounded away from $1/2$ by an inverse polynomial in the size of the input. For simplicity, we just consider being bounded away from $1/2$ by a constant.}. See page 9 of \cite{LB-2014} for a discussion of this and the fact that the acceptance probability can be amplified as in the usual quantum case. Given these definitions, the following theorem was proved in \cite{LB-2014}.

\begin{theorem}  
For any generalised probabilistic theory  $\bold{G}$ satisfying tomographic locality, we have $$\bold{BGP}\subseteq\bold{AWPP}\subseteq\bold{PP}\subseteq\bold{PSPACE}$$ 
\end{theorem}

It is worth noting that due to the computation of the acceptance of an input $x$, we are given polynomial deterministic classical computation ``for free''. As a result, the lower bound of $\bold{P}\subseteq\bold{BGP}$ is satisfied for all theories $\bold{G}$.

One can define a notion of generalised circuits with the ability to \emph{post-select} on at most exponentially-unlikely circuit outcomes. These are poly-sized uniform circuits in a general theory, where the probability of acceptance is conditioned on the circuit outcome $z$ lying in a (poly-time computable) subset of all possible values of $z$. 

\begin{definition}
A language $\mathcal{L}$ is in the class $\bold{PostBGP}$ if there is a poly-sized uniform circuit family in that theory and an efficient acceptance condition, such that 
\begin{enumerate}
\item There exists a constant $D$ and polynomial $w$ such that $P(z\in S)\geq 1/D^{w(|x|)}$
\item If $x\in\mathcal{L}$ then $P_x(\mathrm{accept}|z\in S)\geq\frac{2}{3}$
\item If $x\notin\mathcal{L}$ then $P_x(\mathrm{accept}|z\in S)\leq\frac{1}{3}$
\end{enumerate}  
where $z$ is the circuit outcome, $S$ is a subset of all possible circuit outcomes and $z\in S$ can be checked by a Turing machine in time polynomial in $|x|$.
\end{definition}
Aaronson showed in \cite{Post} that $\bold{PostBQP}=\bold{PP}$ and the following theorem was shown in \cite{LB-2014}.
\begin{theorem} \label{post}
For any generalised probabilistic theory $\bold{G}$ satisfying tomographic locality, we have $$\bold{PostBGP}\subseteq\bold{PostBQP}$$
\end{theorem}

\section{Proofs and advice} \label{general-def}  

In this section we provide generalisations of the definitions of classical (quantum) computing with advice and a type of classical (quantum) interactive proof system to the framework of general operational theories. For an overview of the classical (quantum) definitions see appendix \ref{app} (appendix \ref{app1}). We will assume the reader is familiar with the definition of the familiar complexity classes $\bold{P}$ and $\bold{NP}$ as well as the formalism of quantum circuits. As with the definition of $\bold{BGP}$, unless otherwise stated, the constants $(\frac{2}{3},\frac{1}{3})$ can be chosen arbitrarily as long as they are bounded away from $\frac{1}{2}$ by some constant. 

\subsection{Definitions for general theories}  

Circuits from a uniform circuit family $\{C_x\}$ in some general theory are indexed by the string $x$ that encodes the decision problem the theory is attempting to solve. In defining the class of efficient computation in a theory, the family $\{C_x\}$ is taken to consist of closed circuits from that theory. This will not be the case when advice and proofs are involved; in this paradigm, one is given both the problem instance $x$ and a proof or advice state, so the constructed circuit $C_x$ must have open system ports into which this state can be plugged. Henceforth we will assume that uniform circuit families consist of collections of circuits with a number of open input ports, which can grow as a polynomial in $|x|$, which we call the \emph{auxiliary register}. Note that the choice of finite gate set determines the possible system types of the auxiliary register. Given this convention, we can define efficient computation with trusted advice in a specific general theory.


\begin{definition}
For a general theory $\bold{G}$, a language $\mathcal{L}\subseteq\{0,1\}^{n}$ is in the class $\bold{BGP/gpoly}$ if there exists a poly-sized uniform family of circuits $\{C_x\}$ in $\bold{G}$, a set of (possibly non-uniform) states $\{\sigma_{|x|}\}_{n\geq 1}$ on a composite system of size $d(n)$ for some polynomial $d:\mathbb{N}\rightarrow\mathbb{N}$, and an efficient acceptance criterion, such that for all strings $x\in\{0,1\}^n$:
\begin{enumerate}
\item If $x\in\mathcal{L}$ then $C_x$ accepts with probability at least $2/3$ given $\sigma_{n}$ as input to the auxiliary register.

\item If $x\notin\mathcal{L}$ then $C_x$ accepts with probability at most $1/3$ given $\sigma_{n}$ as input to the auxiliary register.
\end{enumerate}
\end{definition}

Here by ``composite system of size $d(n)$'', we mean that the number of systems, or open ports, of the auxiliary register -- into which the advice state is input -- increases as $d(n)$, for $d$ a polynomial in the input size. Since, as mentioned, there an efficient, deterministic classical computer deciding acceptance and each state $\sigma_{n}$ has a classical pointer associated with it, classical advice can always be encoded into these pointers (of which there can be polynomially many). Therefore, we can always give the lower bound $\bold{P/poly}\subseteq\bold{BGP/poly}\subseteq\bold{BGP/gpoly}$, where the suffix $\bold{/poly}$ denotes classical advice.

\begin{definition}
For a general theory $\bold{G}$, a language $\mathcal{L}\subseteq\{0,1\}^{n}$ is in the class $\bold{GMA}$ if there exists a poly-sized uniform family of circuits $\{C_x\}$ in $\bold{G}$, a polynomial $d:\mathbb{N}\rightarrow\mathbb{N}$ and an efficient acceptance criterion, such that for all strings $x\in\{0,1\}^n$:
\begin{enumerate}
\item If $x\in\mathcal{L}$ then there exists a (possibly non-uniform) proof state $\sigma$ on a composite system of size $d(n)$ such that $C_x$ accepts with probability at least $2/3$ given $\sigma$ as input to the auxiliary register.

\item If $x\notin\mathcal{L}$ then $C_x$ accepts with probability at most $1/3$ given $\sigma$ as input to the auxiliary system, for all states $\sigma$.
\end{enumerate}
\end{definition}

We refer the reader to Appendix \ref{app} for the definitions of computation with advice and proofs in the case of classical and quantum theory as we will make reference to these complexity classes throughout the paper. Informally, for the specific case of quantum theory, the $\bold{G}$ in the nomenclature should be replaced with $\bold{Q}$ and $\bold{/gpoly}$ is replaced with $\bold{/qpoly}$.

The existential quantifiers in the above definition of $\bold{GMA}$ rigorously capture the notion of a circuit having to ``verify'' the proof. Note also that advice states can only depend on the size of the input whereas proofs can, in general, be dependent on the inputs themselves. The amplification procedure of \cite{Kitaev-Watrous} that achieves exponential separation for the acceptance and rejection probabilities in $\bold{QMA}$, at the expense of a polynomial increase in the size of the witness state, can be adapted in a straightforward fashion to provide a similar amplification procedure for $\bold{GMA}$, for arbitrary $\bold{G}$. Note that $\bold{BGP}\subseteq\bold{GMA}$ follows straightforwardly from the definitions. Also, via the same arguments given to lower bound the class $\bold{BGP/gpoly}$, we can always give the lower bound $\bold{NP}\subseteq\bold{GMA}$.

It was proved in \cite{Kitaev-Watrous} that $\bold{QMA}\subseteq\bold{PP}$, and this was improved in \cite{APP} to $\bold{QMA}\subseteq\bold{A_0PP}$, (see also \cite{mariott-watrous}). Aaronson and Drucker \cite{Aaronson-Drucker} have shown the following remarkable relation between these two classes: $$\bold{BQP/qpoly}\subseteq\bold{QMA/poly}.$$ This says that one can always replace (poly-size) quantum advice by (poly-size) classical advice, together with a (poly-size) quantum proof\footnote{Note that advice can encode solutions to even undecidable problems, any upper bound on an advice class will be another advice class.}. Intuitively, this relation can be summed up as follows: one can always simulate an arbitrary quantum state $\rho$ on all small circuits, using a different state $\widetilde{\rho}$ that is easy to recognize\footnote{One can even take $\widetilde{\rho}$ to be the ground state of a local Hamiltonian \cite{Aaronson-Drucker}.}. In Section~\ref{trade-off} we investigate whether this relation holds for general operational theories.

\subsection{Example: Boxworld}

We now look at Boxworld with respect to our definitions of proofs and advice in general physical theories. Towards this end we provide a brief definition of Boxworld, see e.g. \cite{PR-trade-off} for a more in-depth discussion. For a given single system $A$ in Boxworld, there are two choices of binary-outcome measurements, $\{_A(x_a|\}$ for $x,a\in\{0,1\}$. Here $x$ is the bit denoting the two possible choices of measurement and $a$ is the bit denoting the two possible outcomes of the chosen measurement, i.e the two measurements on system $A$ are $\{_A(0_0|, _A(0_1|\}$ and $\{_A(1_0|, _A(1_1|\}$. States and measurements in this theory can produce correlations associated with the so-called Popescu-Rohrlich non-local box \cite{PRbox}. That is, for a bipartite system $AB$, there exist states $|\rho_{PR})_{AB}$ such that
$$(x_a|(y_b|\rho_{PR})_{AB}= \left\{
  \begin{array}{lr}
   \frac{1}{2}, \ \mathrm{if} \ a\oplus b= xy,\\
    0, \ \mathrm{otherwise}
  \end{array}
\right.
$$
where $\oplus$ represents addition modulo $2$. These correlations can be extended to an $n$-partite system where now for the $j$th party, $x_{j}\in\{0,1\}$ and $a_{j}\in\{0,1\}$ are the choice of measurement and its outcome respectively. There exists a state $|\rho_{f})$ and effects $\{_j(x_{j},a_{j}|\}$ for all $j$ parties that produce the probabilities \cite{Barrett-Pironio,Hoban-2011}
$$(x_{1},a_{1}|(x_{2},a_{2}|...(x_{n},a_{n}|\rho_f)= \left\{
  \begin{array}{lr}
   \frac{1}{2^{n-1}}, \ \mathrm{if} \ \bigoplus_{j=1}^{n}a_{j}= f(x),\\
    0, \ \mathrm{otherwise}
  \end{array}
\right.
$$
where $\bigoplus$ represents summation modulo $2$ and $f:\{0,1\}^{n}\rightarrow\{0,1\}$ is any Boolean function from the bit-string $x$ with elements $x_{j}$. Therefore, if the state $|\rho_{f})$ is prepared and local measurements described by effects $(x_{j},a_{j}|$ made, a classical computer can compute the parity of all outcomes $a_{j}$ and so we deterministically obtain the evaluation of Boolean function $f(x)$. This relatively straightforward observation gives us the following result.

\begin{theorem}\label{thmboxworld}
There exists a generalised probabilistic theory  $\bold{G}$ satisfying causality and tomographic locality, which satisfies $\bold{BGP/gpoly}=\bold{ALL}$ where $\bold{ALL}$ is the class of all decision problems.
\end{theorem}

\begin{proof}

Clearly $\bold{BGP/gpoly}\subseteq\bold{ALL}$ is trivially true for Boxworld. The states $|\rho_{f})$ can be used as advice states and, as all decision problems can be represented by Boolean functions, it follows that $\bold{ALL}\subseteq\bold{BGP/gpoly}$.
\end{proof}


Note that the above proof still goes through if we insist that Boxworld only has reversible dynamics since the proof only requires the ability to prepare and measure states. If one considers the class $\bold{GMA}$ for Boxworld with only reversible transformations then we have $\bold{GMA}\subseteq\bold{MA}$ since all reversible dynamics are trivial in this theory and can thus be simulated classically \cite{Rev,Rev1,Rev2}. By trivial, we mean that the circuits in Boxworld only consist of making the local ``fiducial" measurements $\{_j(x_{j},a_{j}|\}$ on a state and performing classical post-processing on the outcomes. This process can be simulated by the prover giving the verifier the classical string of measurement outcomes similar to the approach of Lemma 2 in \cite{maq}. That is, while poly-size advice states in Boxworld can encode any Boolean function, the theory has no non-trivial dynamics to efficiently verify this function is encoded in the state if the prover cannot be trusted.



\section{Consequences of non-trivial dynamics for computation} \label{dynamics}


In part \ref{a} of this section, we show the existence of non-trivial dynamics implies that computation in that theory is at least as powerful as probabilistic classical computation: $\bold{BPP}\subseteq\bold{BGP}$. Hence non-trivial dynamics imply non-trivial computational power. Furthermore in part \ref{b}, we show the existence of non-trivial dynamics implies a bound on the amount of ``useful'' information -- quantified by the size of the class $\bold{BGP/gpoly}$ -- that can be stored in general states.  

\subsection{Powerful computation from non-trivial dynamics} \label{a}

\begin{definition}
A theory is said to be \emph{non-classical} if, for at least one $n$-tuple of pure and perfectly distinguishable states $\{|\sigma_i)\}_{i=1}^N$, there exists a pure state $|y)$ such that $(e_i|y)=p_i$ for $0<p_i<1$ for all $i$, where $\{(e_i|\}_{i=1}^N$ is the measurement that distinguishes the $\{|\sigma_i)\}_{i=1}^N$. 
\end{definition}

Before we present our result, we wish to emphasize that the result is highlighting the \textit{intrinsic} computational power in a theory. As previously mentioned, in our framework we already have a classical computer that processes experimental data and, if a circuit in a theory $\bold{G}$ can produce random numbers, we can easily achieve the complexity class $\bold{BPP}$. By talking about intrinsic computational power, we imagine reducing the power of our classical computer to perform extremely simple, non-universal classical computation. For example, the classical computer in deciding the output of the computation could only output the classical counter value on one of the measurements. Our result then shows that theories with a certain amount of non-trivial dynamics still decide any problem in $\bold{BPP}$.
\begin{theorem} \label{power}
Let $\bold{G}$ be a causal, non-classical theory with at least two pure and distinguishable states that satisfies Permutability. 
Then $\bold{BPP}\subseteq\bold{BGP}$.
\end{theorem}

\begin{proof}

For $\bold{BPP}\subseteq\bold{BGP}$, it is sufficient to show two things: that transformations of the general theory can simulate the action of any reversible Boolean function $f:\{0,1\}^n\rightarrow\{0,1\}^n$, and that it is possible to prepare a source of random bits. First, bit strings $x=x_1 \ldots x_n$ can be represented by perfectly distinguishable pure states $|x) = |x_1)\otimes \cdots \otimes |x_n)$. Then, the first condition follows from Permutability: since  $\{|f(0\ldots 0)), \dots, |f(1\ldots 1))\}$ is a permutation of the tuple of pure and perfectly distinguishable states $\{|0\ldots0), \dots, |1\ldots 1)\}$, there must exist a reversible transformation $T_f$ such that $T_f|x)=|f(x))$. 

For the second condition it suffices if there are circuits that can generate random bits. Consider the two pure and perfectly distinguishable states $|0)$ and $|1)$. Let $\{(e_0|,(e_1|\}$ be a measurement that distinguishes them, that is $(e_i|j)=\delta_{ij}$, for ${i,j}=0,1$. Non-classicality implies that there exists some pure state $|y)\notin\{|0),|1)\}$ such that $(e_0|y)=p$ and $(e_1|y)=1-p$, with $0<p<1$. Probabilities of $1/2$ can be generated by preparing two copies of $|y)$, implementing the measurement on each in parallel and assigning a value $y=0$ or $1$ to the outcomes $01$ and $10$ respectively\footnote{This argument is based on von Neumann's argument for turning two copies of a biased coin into one copy of an unbiased coin.}. 


\end{proof} 

\subsection{Bounds on computation with advice in physical theories} \label{b}

Recall that a state is \emph{mixed} if it is not pure and it is \emph{completely mixed} if any other state refines it. That is, $|c)$ is completely mixed if for any other state $|\rho)$, there exists a non-zero probability $p$ such that $p|\rho)$ refines $|c)$. Intuitively, one should be able to efficiently prepare a completely mixed state on a computer in any general theory. This follows because the completely mixed state can be prepared by performing any uniform state preparation and ``forgetting'' the outcome. Henceforth we shall assume that the completely mixed state -- if it exists -- is uniform.

Recall the definition of a bit-symmetry from section \ref{gpt}. In any bit-symmetric theory with at least two pure and distinguishable states, it can be shown \cite{MU} that the group of reversible transformations acts \emph{transitively} on the set of pure states. That is, given any two pure states $|\rho), |\sigma)$, there exists a reversible transformation $T$ such that $T|\rho)=|\sigma)$. This fact can be used \cite{Pavia1, Pavia2} to prove the existence of a completely mixed state as the unique state -- for a given system type -- that is invariant under all reversible transformations. 
 
Bit-symmetry is a powerful principle and has many useful consequences. Two more of which are: 
\begin{enumerate}
\item Every bit-symmetric theory is \emph{self-dual} \cite{MU}. That is, to every pure state $|\rho_e)$ there is associated a unique pure effect $(e_\rho|$, and vice versa \footnote{The proof of this fact requires two further technical assumptions, both implicit in Ref.~\cite{MU}. These are: the group of reversible transformations must be compact, and every mathematically allowed effect is physical. }. This association is achieved via an inner product $[.,.]$, on the real vector space $V$ generated by the set of states, as: $(e_\rho|\sigma)=[|\rho_e),|\sigma)]$, for all states $|\sigma)$. Note that $[|\rho),|\rho)]=1$ for all pure states $|\rho)$.
\item Let $\Vert |v) \Vert_{phy}=2\max_{(e|} |(e|v)|$ and $\Vert v \Vert_E=\sqrt{[v,v]}$, for $v$ an arbitrary vector in $V$. The norm $\Vert |\rho)-|\sigma) \Vert_{phy}$ has a natural operational interpretation as the distinguishably of $|\rho)$ and $|\sigma)$. Bit-symmetry implies \cite{black-holes} that $\Vert |\rho)-|\sigma)\Vert_{phy} \leq c \Vert |\rho)-|\sigma)\Vert_E,$ where $c=\Vert |c) \Vert_E$ for $|c)$ the completely mixed state. 
\end{enumerate}

Using the above facts, we now prove a version of the ``as good as new lemma''\footnote{Also called the ``gentle measurement lemma'', which was independently proved by Winter in \cite{Gentle1} and improved by Ogawa and Nagaoka in \cite{Gentle2}} -- discussed in the quantum case in \cite{Aaronson-advice} -- for all bit-symmetric theories. Before we state this lemma, we need to briefly introduce a notion of post-measurement state update rule for bit-symmetric theories. In this work applying a measurement to a state corresponds to a closed circuit -- that is a probability. However, to discuss post-measurement states, this must be generalised slightly. A measurement will henceforth correspond to a laboratory device from some input state to the output post-measurement state, where the classical pointer denotes the outcome of the measurement. Consider the measurement $\{(i|\}$, consisting of pure effects $(i|$, and apply it to some state $|\rho)$. On observing outcome $i$, the state $|\rho)$ is updated to $|\rho_i)/(u|\rho_i)$ where $|\rho_i)$ is the unique pure state associated to $(i|$. This state update rule satisfies a natural \emph{repeatability} condition: any state yielding outcome $i$ with unit probability is left invariant by the update rule, thus repeated measurements always yield the same result. See \cite{P-W} for more in-depth discussion of state update rules in general theories.

\begin{lemma} \label{as good as new}
Given a  two outcome measurement, consisting of the pure effects $\{(0|, (1|\}$, and a state $|\rho)$ such that $(0|\rho)=1-\epsilon$, for $\epsilon\geq 0$, the post-measurement state on observing outcome $0$ satisfies
$$\Vert |\rho)-|\rho_0) \Vert_{phy} \leq c\sqrt{2\epsilon},$$
where $c=\Vert |c) \Vert_E$ is the completely mixed state, in all bit-symmetric theories.
\end{lemma}
\begin{proof}
Recall in a bit-symmetric theory that $\Vert |\rho)-|\sigma)\Vert_{phy} \leq c \Vert |\rho)-|\sigma)\Vert_E.$ We thus have
$$\begin{aligned}
\Vert |\rho)-|\rho_0)\Vert_{phy} &\leq c \Vert |\rho)-|\rho_0)\Vert_E \\
&= c\sqrt{[|\rho)-|\rho_0), |\rho)-|\rho_0)]} \\
& \leq c\sqrt{2-[|\rho),|\rho_0)]-[|\rho_0),|\rho)] } \\
&= c\sqrt{2-2[|\rho_0),|\rho)]} \\
&=c\sqrt{2-2(0|\rho)}=c\sqrt{2\epsilon}.
\end{aligned}
$$
The first line follows from the definition of $\Vert . \Vert_E$, the second from the fact that $\Vert |\sigma) \Vert_E\leq 1$ for all $|\sigma)$, the third from the symmetry of the inner product $[.,.]$ and the last from the definition of self-duality.
\end{proof}

The above lemma states that if one outcome of a two-outcome measurement occurs with high probability on some state, then the post-measurement state after getting that outcome is ``close'' to the original state. We are now in a position to state the main result of this section. Before we do, let us fix the accepting criterion for computation with advice and make the simplifying assumption that the accept/reject measurement consists of pure effects.

\begin{theorem}\label{advicetheorem}
Any causal, bit-symmetric, tomographically local theory $\bold{G}$ with at least two pure and distinguishable states satisfies
$$\bold{BGP/gpoly} \subseteq \bold{PostBGP/poly} \subseteq \bold{PP/poly}$$
\end{theorem}

The above theorem states that in theories with non-trivial dynamics, there is a bound to how much useful information one an extract from any state. This result provides evidence for the existence of a trade-off between states and dynamics and can be seen as a natural converse to the results of \cite{Rev,Rev1,Rev2}. Our proof is a slight variation of the original proof in the quantum case, due to Aaronson \cite{Aaronson-advice}.

\begin{proof}
Begin by amplifying the success probability of $\bold{BGP/gpoly}$ on input $x$ from $2/3$ to $1-1/2^{q(|x|)}$. This is achieved by running a polynomial number of copies of the circuit $C_x$ in parallel and taking the majority answer. Note that in this amplification scheme the total advice state is the (vector space) tensor product of advice states for each individual circuit. Recall that the completely mixed state $|c)$ is assumed to satisfy uniformity and that there exists a non-zero probability $p$ such that $p|\sigma)$ is a refinement of $|c)$, for any $|\sigma)$. Uniformity implies that $p$ can be well approximated by some rational $c/d^{w(|x|)}$, for $c$ an integer and $d$ a polynomial in the size of the input $x$ (see the proof of Theorem $14$ in appendix B of \cite{LB-2014} for a more in-depth discussion of uniformity).  
 
Given any language $\mathcal{L}\in\bold{BGP/gpoly}$ we now construct a $\bold{PostBGP/poly}$ algorithm that decides $\mathcal{L}$. Given some $x$, use the completely mixed state as the advice to the circuit $C_x$. Now, from the definition of $\bold{BGP/gpoly}$, if $|c)$ cannot be used as advice to determine $x\in\mathcal{L}$, the circuit accepts with probability less than $1/3$. Consider the the post-measurement state $|c')$ of the auxiliary register after running $C_x$ with advice $|c)$ \emph{post-selecting} on the event that we succeeded in outputting the correct answer. If $|c')$ cannot be used as advice for all inputs, there exists some $x'$ such that $C_{x'}$ succeeds with probability less than $1/3$. As before, consider the post-measurement state of the auxiliary register after running $C_{x'}$ with advice $|c')$ post-selecting on outputting the correct answer. Continue in this fashion for some $t(|x|)$ stages, $t$ a polynomial. Successful post-selection is guaranteed as the actual advice state refines $|c)$ with probability $c/d^{w(|x|)}$. 

If, at any iteration of this process, we cannot find an $x$ to move forward, we must be holding a state that works as advice for every input, and we can use it to run $C_z$ on any input $z$, succeeding with high probability. Thus if the process halts after a polynomial number of iterations, we are done.


If the correct advice state $|\sigma)$ had been used in the computation, Lemma~(\ref{as good as new}) would imply the post-measurement state on observing the accept outcome, $|\sigma_{acc})$, would -- under the simplifying assumption that the accept/reject measurement consists of pure effects -- satisfy:
$$ \Vert |\sigma)-|\sigma_{acc}) \Vert_{phy} \leq c\sqrt{\frac{1}{2^{q(|x|)-1}}}. $$
As the completely mixed state $|c)$ is uniform, it follows that $c=\Vert |c) \Vert_E\leq O(2^{m(|x|)})$ for $m$ a polynomial. Therefore, $c/\sqrt{2^{q(|x|)-1}}=o(1)$. We thus have  
$$ \Vert |\sigma)-|\sigma_{acc}) \Vert_{phy} \leq o(1).$$
Therefore on each iteration of the above process, the correct answer is output with probability $$\frac{c}{d^{w(|x|)}}\left(1-o(1)\right).$$

This process has been designed so that the probability that $|c)$ can be re-used on each iteration and succeed at each stage is at most $1/3^{t(|x|)}$. Therefore, we have that $$\frac{c}{d^{w(|x|)}}\left(1-o(1)\right) \leq 1/3^{t(|x|)}.$$ Thus $t(|x|) \leq O(w(|x|))$ and we are done.
 
There thus exists a polynomial number of $x_1,\dots,x_t$ such that, if $|a)$ is the post-measurement state after we start with $|c)$ and post-select on succeeding on each $x_i$ in turn, $|a)$ is a good advice state for every string $z$. Provide the algorithm with this sequence of classical strings, along with the correct outcomes $b,\dots,b_t$ for each of them. The algorithm then prepares $|c)$, uses it as advice and post-selects on getting outcomes $b,\dots,b_t$. After this process we obtain the state $|a)$ and so all languages that can be decided in $\bold{BGP/gpoly}$ can also be decided in $\bold{PostBGP/poly}$ and thus, by tomographic locality and Theorem \ref{post}, in $\bold{PP/poly}$.
\end{proof}

\section{Bounds on the power of proofs in physical theories} \label{trade-off} 

In this section we will put a non-trivial bound on $\bold{GMA}$. To state our result, the notion of a $\bold{GapP}$ function must be introduced. Given a poly-time non-deterministic Turing Machine $n$ and input string $x$, let $N_{acc}(x)$ be the number of accepting computation paths of $N$ given input $x$, and $N_{rej}(x)$ the number of rejecting computation paths of $N$ given $x$. A function $f:\{0,1\}^{*}\rightarrow\mathbb{Z}$ is a $\bold{GapP}$ function if there exists a polynomial-time non-deterministic Turing Machine $N$ such that $f(x)=N_{acc}(x)-N_{rej}(x)$, for all input strings $x$. We can now define the class $\bold{A_0PP}$.

\begin{definition}
A language $\mathcal{L}$ is in the class $\bold{A_0PP}$ if and only if there exists a $\bold{GapP}$ function $f$ and an efficiently computable function $T$ such that
\begin{enumerate}
\item for all $x\in\mathcal{L}$ $f(x) \geq T(x)$ and;

\item for all $x\notin\mathcal{L}$ we have $0 \leq f(x) \leq \frac{1}{2}T(x)$ 
\end{enumerate}
\end{definition}

It has been shown that the above class is contained in $\bold{PP}$. Fix the efficient acceptance condition for proof verification so that, in all uniform circuits, the measurement applied at the end of the computation to the auxiliary register consists of only unit effects. We make this choice to move closer to the standard quantum acceptance condition. We also make the simplifying assumption -- routinely made in the literature -- that all mathematically allowed states are physically allowed. That is, all vectors whose inner product with any effect is in $[0,1]$ correspond to physical states.

\begin{theorem} \label{main}
For any generalised probabilistic theory  $\bold{G}$ satisfying causality, tomographic locality and the assumption that all mathematically allowed states are physically allowed, we have that $$\bold{GMA}\subseteq\bold{A_0PP}\subseteq\bold{PP}.$$ 
\end{theorem}

\begin{proof}

Recall that any matrix $M$ has a singular value decomposition given by $M=UDV^T$, where $U,V$ are unitary (orthogonal if the matrix is real) matrices, $V^T$ is the transpose of $V$ and $D$ is a diagonal matrix. The diagonal entries of $D$ are all non-negative real numbers and are called the \emph{singular values} of the matrix $M$. Note that the eigenvalues of the matrix $M^TM=VD^TDV^T$ are the squares of the singular values of $M$. 

Let $M_x$ be the matrix representation of the uniform circuit, including states and effects on the non-auxiliary register, on input $x$, $(u|$ be the (tensor product) of unit effects applied on the auxiliary register and $|\rho)$ be any arbitrary state (which can be non-uniform) input to the auxiliary register. Without loss of generality, one can pad this matrix (and row and column vector) with rows and columns of zeros to ensure it is square. The probability that the circuit accepts the string $x$ is given by $(u|M_x|\rho)$. It will now be shown that this probability is upper bounded by the largest singular value of the matrix $M_x$. Consider the following
$$(u|M_x|\rho)=(u|UDV^T|\rho) \leq \sigma_{max} (u|UV^T|\rho),$$
where $\sigma_{max}$ is the largest singular value of $M_x$. Now $UV^T$ is a unitary matrix and so can be decomposed as follows $UV^T=WD'W^T$, where $W$ is another unitary matrix and $D'$ is a diagonal matrix consisting of the eigenvalues of $UV^T$, recall that these eigenvalues all have absolute value $1$. Thus,
$$ (u|M_x|\rho) \leq \sigma_{max} (u|WD'W^T|\rho) \leq \sigma_{max} (u|\rho) \leq \sigma_{max},$$
where the second inequality follows from that fact that the entries of $D'$ have absolute value $1$ and that $W$ is unitary and the third inequality follows as $(u|\rho) \leq 1$.

Now as the squares of the singular values are the eigenvalues of the (positive definite) matrix $M_x^TM_x$, we have that 
$$(\sigma_{max}^2)^d \leq \mathrm{Tr}\big((M_x^TM_x)^d\big) \leq 2^n (\sigma_{max}^2)^d,$$
where $2^n$ is the number of entries on the diagonal of $M_x^TM_x$, $n$ is a polynomial in $|x|$ an $d$ is an arbitrary natural number. Let $d$ be a polynomial in $|x|$ that takes values in the natural numbers and assume without loss of generality that it grows faster than the polynomial $n$, we will need this requirement later. 

The matrix $M_x$ satisfies the uniformity condition, and it was shown in \cite{LB-2014} that the entries of all such matrices are $\bold{GapP}$ functions. By the closure properties of $\bold{GapP}$ (again see appendix B of \cite{LB-2014}) functions the entries in the matrix $(M_x^TM_x)^d$ are also $\bold{GapP}$ functions. Using an argument similar to that in \cite{APP}, $\mathrm{Tr}\big((M_x^TM_x)^d\big)$ can be straightforwardly shown to be a $\bold{GapP}$ function, denote it by $f(x)$. So, from the definition of $\bold{GMA}$, we have that $f(x) \geq \sigma_{max}^{2d} \geq \big(\frac{2}{3}\big)^{2d}$ for all $x$ in the language. 

Now the vector that achieves the bound of $\sigma_{max}$ is the right singular vector of $M_x$ with singular value $\sigma_{max}$, which we denote by $|\sigma)$. If this vector is a physical state then we are done, as it follows from the definition of $\bold{GMA}$ and an argument similar to the one above that $f(x) \leq \frac{1}{2}\big(\frac{2}{3}\big)^{2d}$ for all $x$ not in the language. If this vector is not a physical state then we have a bit more work to do. 

Towards this end, consider the following. We are free to re-parametrise (e.g. see page $7$ of \cite{ML}) the set of states by an affine transformation $\phi: \mathbb{R}^m\rightarrow\mathbb{R}^m$, where $\mathbb{R}^m$ is the (smallest) real vector space that contains the set of states, as follows: 
$$\begin{aligned} &|\rho)\rightarrow|\widetilde{\rho})=\phi|\rho), \quad (a|\rightarrow(\widetilde{a}|=(a| \phi^{-1} \\  &\quad \mathrm{and,} \ \quad M_x\rightarrow\widetilde{M_x}=\phi M_x \phi^{-1} \end{aligned}$$
as this does not change the probabilities, i.e. $(a|M_x|\rho)=(\widetilde{a}|\widetilde{M_x}|\widetilde{\rho})$. Now, as an affine transformation is just a translation followed by a scaling, choose $\phi$ so that the Euclidean unit ball is contained in the re-parametrised state space (just translate the original state space and scale it appropriately to ensure this, noting that translations and scaling are reversible). As the singular vectors of every matrix are unit vectors, without loss of generality they are contained in this unit ball. Under the assumption that all mathematically allowed states are physically allowed ensures these singular vectors are physical states. Thus $\sigma_{max}=(\widetilde{u}|\widetilde{M_x}|\sigma),$  
where $(\widetilde{u}|$ is the unique deterministic effect. The causality principle ensures that any state $|\widetilde{s})$ can be scaled so that $(\widetilde{u}|\widetilde{s})=1$, see e.g. \cite{Pavia1}.  So for $x$ not in the language we have $\sigma_{max}=(\widetilde{u}|\widetilde{M_x}|\sigma)\leq 1/3$.

It follows that    
$$f(x) \leq 2^n \sigma_{max}^{2d} \leq 2^n \Big(\frac{1}{3}\Big)^{2d} \leq \frac{1}{2}\Big(\frac{2}{3}\Big)^{2d},$$
where the first inequality follows from $\mathrm{Tr}\big((M_x^TM_x)^d\big) \leq 2^n (\sigma_{max}^2)^d$ and the last inequality follows from the fact that, for $d$ increasing sufficiently faster than $n$, we have $2^{n+1}\leq 4^d$.  

Thus, for a language $\mathcal{L}$ in $\bold{GMA}$ we have    
\begin{enumerate}
\item for all $x\in\mathcal{L}$ there exists a $\bold{GapP}$ function $f$ such that $f(x) \geq \big(\frac{2}{3}\big)^{2d}$ and;

\item for all $x\notin\mathcal{L}$ we have $f(x) \leq \frac{1}{2}\Big(\frac{2}{3}\Big)^{2d},$
\end{enumerate}
and so we have that $\bold{GMA}\subseteq\bold{A_0PP}$.  

\end{proof}

\section{Relating proofs and advice?}  \label{optimal} 


The following relation, discussed in section \ref{general-def}, $\bold{BQP/qpoly}\subseteq\bold{QMA/poly},$ captures an intriguing feature of proofs and advice in quantum theory: one can always replace quantum advice with classical advice together with a quantum proof. Here we study the relation
\begin{equation}\label{desirable}
\bold{BGP/gpoly}\subseteq\bold{GMA/poly},
\end{equation}
in general theories. 
Note that the relation is satisfied in classical computation:
\begin{equation}
\bold{BPP/rpoly}=\bold{P/poly}\subseteq\bold{NP/poly}\subseteq\bold{MA/poly},\nonumber
\end{equation}
where $\bold{BPP/rpoly}=\bold{P/poly}$ was shown in \cite{de, Adl78}. Clearly the relation in \eqref{desirable} is then not uniquely satisfied by quantum theory, but one could ask whether quantum theory is the most computationally powerful theory in which \eqref{desirable} is satisfied? 

Using these observations as motivation we obtain the following corollary of Theorem~\ref{main}.
\begin{Corollary}
There exist general theories $\bold{G}$ satisfying tomographic locality and causality such that
$\bold{BGP/gpoly}\nsubseteq\bold{GMA/poly}$.
\end{Corollary}
\begin{proof}
Firstly, we can use Theorem~(\ref{main}) to conclude that $\bold{GMA/poly}\subseteq\bold{PP/poly}$ and by a counting argument $\bold{PP/poly}$ is strictly contained in $\bold{ALL}$. From Theorem~(\ref{thmboxworld}), there exists a theory $\bold{G}$ such that $\bold{ALL}=\bold{BGP/gpoly}$ and so we do not have $\bold{BGP/gpoly}\subseteq\bold{GMA/poly}\subseteq\bold{PP/poly}$ for this theory. 
\end{proof}

Motivated by the above corollary we can say something non-trivial about theories where $\bold{BGP/gpoly}\nsubseteq\bold{GMA/poly}$. Consider the case of using a polynomially-sized circuit from a specific theory, built from any fixed gate set in that theory, to prepare an arbitrary, but polynomially large, state in the theory.  Given this set-up, we can prove the following result. 
\begin{theorem}
In any general theory $\bold{G}$ with $$\bold{BGP/gpoly}\nsubseteq\bold{GMA/poly}$$ there exist states (of polynomial size) that cannot be prepared using an efficient circuit built from any gate set in the theory.
\end{theorem}

\begin{proof}
Assume toward contradiction that all states can be prepared using an efficient circuit built from any gate set in the theory. Thus, as there must exist a classical description of each circuit, any advice state from this theory can be replaced with the \emph{classical} advice that specifies the description of the circuit that efficiently prepares the given advice state. We thus have
$$\bold{BGP/gpoly}\subseteq\bold{BGP/poly}\subseteq\bold{GMA/poly},$$ which is a contradiction. There must therefore exist at least one state that cannot be prepared efficiently in this theory. 
\end{proof}

Thus in theories that do not satisfy $$\bold{BGP/gpoly}\subseteq\bold{GMA/poly},$$ the dynamics are not rich enough to prepare the states that contain a large amount of ``useful'' information. This is not to say that in theories satisfying this relation every state can be efficiently prepared, it is just that in theories violating the relation this assertion can be proved \emph{directly} from the violation. As a side remark, within the theorem proof we have proven that $\bold{BGP/poly}$ is \textit{strictly} contained in $\bold{BGP/gpoly}$ for theories $\bold{G}$ where $\bold{BGP/gpoly}\nsubseteq\bold{GMA/poly}$. It is presently unknown if quantum advice is strictly stronger than classical advice for quantum computers.

In addition to proving $\bold{BGP/gpoly}\subseteq\bold{GMA/poly}$, Aaronson and Drucker proved what they called a ``Quantum Karp-Lipton" theorem \cite{Aaronson-Drucker}. The Karp-Lipton theorem states that if $\bold{NP}\subseteq\bold{P/poly}$ then the polynomial hierarchy collapses to its second level, which is believed to be unlikely \cite{Karp-Lipton}. The Quantum Karp-Lipton theorem states that if $\bold{NP}\subseteq\bold{BQP/qpoly}$ then the second level of the polynomial hierarchy is contained in $\bold{QMA}^{\bold{PromiseQMA}}$ \footnote{Here $\bold{PromiseQMA}$ is the same as $\bold{QMA}$ except there is a ``promise" on the inputs, i.e. all the inputs satisfy some property.}, which is also thought to be unlikely \cite{Aaronson-Drucker}. We refer the reader to the original works for further details but we only wish to highlight that, due to Theorem~(\ref{thmboxworld}), there exist theories $\bold{G}$ where $\bold{NP}\subseteq\bold{BGP/gpoly}$ is necessarily satisfied. Therefore, we cannot obtain a ``Generalised Karp-Lipton" theorem where unlikely consequences are expected from assuming $\bold{NP}\subseteq\bold{BGP/gpoly}$.

\subsection{Related work}

Evidence for the existence of a general trade-off has also appeared in recent work which has considered theories satisfying the no-signalling condition from the point-of-view of interactive proofs. The Merlin-Arthur game is an example of an interactive proof. Another example is a multi-interactive prover ($\bold{MIP}$) system where more than one of these all-powerful provers sends classical bit-strings to a probabilistic classical computer verifier \cite{Ben-Or}. Just as in the Merlin-Arthur game, the provers cannot be trusted. However, these provers are not permitted to communicate with one another. A quantum generalisation of this is to allow the provers to share entangled quantum states. In work by Ito and Vidick \cite{Ito-Vidick}, in this quantum generalisation of $\bold{MIP}$ it is possible for the verifier to efficiently compute problems in the class $\bold{NEXP}$ which is the class of problems evaluated by a non-deterministic computer running in time exponential in the size of the input. However, recent work by Kalai, Raz and Rothblum \cite{NS} has shown that if the provers share resources that satisfy only the no-signalling principle (such as Boxworld), then the problems that can be solved in such a model are actually contained in the class $\bold{EXP}$. Since $\bold{EXP}\subseteq\bold{NEXP}$, in a theory with states more non-local than quantum mechanics these interactive proof systems have $\textit{less}$ computational power, unless $\bold{EXP}=\bold{NEXP}$.

\section{Discussion and conclusion}

The results in this paper provide another example where the best known upper bound on the quantum class $\bold{QMA}$ follows from very minimal assumptions on what constitutes an operational theory. This raises the question of whether better bounds can be derived in the quantum case by exploiting some of the structure unique to quantum theory. 


While the definitions of advice and proof verification presented in this paper can be applied to any theory in the framework, they seem to intuitively encode a notion of causality. Note that in a non-causal theory, circuits do not have any particular ``direction'' and so inputting a given state at the ``start'' of the computation is not the most natural situation one could consider. Instead of receiving an advice \emph{state}, a more natural situation might be to receive an advice \emph{circuit fragment} -- consisting of either a state, transformation or measurement -- which can be plugged into the circuit as it is being built. It would be interesting to determine if this more general definition coincides with the standard one in extensions of quantum theory with indefinite causal structure \cite{PaviaNoCause, BruknerNoCause}.  

On a final note, it would be fascinating to show if the analysis of computation in generalised probabilistic theories could say something concrete about quantum computing. In an analogous fashion, tools from quantum theory have been used to prove results in \textit{classical} computer science, see \cite{Review} for a nice review of such results. We speculate that by understanding quantum theory better within the framework of more general theories we can use tools from the latter to prove results in the former.
\newline

\subsection*{Acknowledgements} 
The authors acknowledge John H. Selby and all his charitable works, and thank Jon Barrett for discussions. CML acknowledges funding from the EPSRC and University College Oxford. MJH acknowledges support from the EPSRC (through the NQIT Quantum Hub) and the FQXi Large Grant \textit{Thermodynamic vs information theoretic entropies in probabilistic theories}.

\begin{appendices} 

\section{Classical case} \label{app}

The study of non-uniform classical computation begins with polynomial-sized Boolean circuits. These circuits can equivalently be viewed as Turing Machines that take polynomial-sized advice bit-strings. These strings only depend on the size of the input and not the input itself. If the string were to depend on every input then we could just encode the solution to any problem for that input and be able to decide any language. The class of decision problems that are solved by a (uniform) deterministic classical computer with classical advice is denoted $\bold{P/poly}$, where the suffix $\bold{/poly}$ denotes a classical advice bit-string.

\begin{definition}
$\bold{P/poly}$ is the class of languages $\mathcal{L}\subseteq\{0,1\}^n$ for which there exists a poly-time uniform classical circuit family $\{\mathcal{C}_x\}$ and a set of bit-strings $\{y_{n}\}_{n\geq 1}$ of length $d(n)$ for some polynomial $d$, such that for all strings $x\in\{0,1\}^n$, $x\in\mathcal{L}$ if and only if $\mathcal{C}_x$ accepts for $(x,y_{n})$ as input.
\end{definition}

Since we will be considering probabilistic processes in full generality, it is worth defining the relevant class of computation with advice where processes are not deterministic. In full generality, we allow the possibility that the advice bit-strings are sampled from a probability distribution for each input size -- we denote such advice as ``randomized advice" denoted by the suffix $\bold{/rpoly}$. In addition to this, we allow the uniform circuits to accept inputs with some error as is normal in efficient probabilistic computation (cf. the definition of $\bold{BGP}$). Therefore the class $\bold{BPP/rpoly}$ of problems solved (with some error) by a (uniform) classical circuit with randomized advice can now be defined.

\begin{definition}
$\bold{BPP/rpoly}$ is the class of languages $\mathcal{L}\subseteq\{0,1\}^n$ for which there exists a poly-time uniform classical circuit family $\{\mathcal{C}_x\}$ and a set of randomized advice bit-strings $\{y_{n}\}_{n\geq 1}$ of length $d(n)$ for some polynomial $d$, such that for all strings $x\in\{0,1\}^n$:
\begin{enumerate}

\item If $x\in\mathcal{L}$ then $\mathcal{C}_x$ accepts with probability at least $2/3$ given $(x,y_{n})$ as input.

\item If $x\notin\mathcal{L}$ then $\mathcal{C}_x$ accepts with probability at most $1/3$ given $(x,y_{n})$.

\end{enumerate}
\end{definition}

Interestingly, despite the ability to use probabilistic processes, via derandomisation arguments it can be shown that $\bold{BPP/rpoly}=\bold{P/poly}$ \cite{Adl78, de}.

In the case where an efficient computer is given a proof from some untrusted provider we have already mentioned the classical complexity class $\bold{NP}$ but this is not the most general class for probabilistic computation. If the efficient classical computer accepts some input with some error, then this is in the remit of Merlin-Arthur games with complexity as follows.

\begin{definition}
$\bold{MA}$ is the class of languages $\mathcal{L}\subseteq\{0,1\}^n$ for which there exists a poly-time uniform classical circuit $\{\mathcal{C}_x\}$ and a polynomial $d$, such that for all strings $x\in\{0,1\}^n$:
\begin{enumerate}

\item If $x\in\mathcal{L}$ then there exists a proof $z\in\{0,1\}^{d(n)}$ such that $\mathcal{C}_x$ accepts with probability at least $2/3$ given $(x,z)$ as input.

\item If $x\notin\mathcal{L}$ then $\mathcal{C}_x$ accepts with probability at most $1/3$ given $(x,z)$ as input, for all proofs $z$.

\end{enumerate}
\end{definition}
The existential quantifiers in the above definition rigorously capture the notion of a circuit having to `verify' the proof. It immediately follows that $\bold{NP}\subseteq\bold{MA}$. This definition will also allow us to readily present the quantum analogue to this class along with its analogue for all possible general theories.

\section{Quantum case} \label{app1}

The class of decision problems that can be solved by an efficient quantum computer with quantum advice, denoted by $\bold{BQP/qpoly}$, is defined as follows.

\begin{definition}
$\bold{BQP/qpoly}$ is the set of languages $\mathcal{L}\subseteq\{0,1\}^n$ for which there exists a poly-time uniform quantum circuit family $\{\mathcal{Q}_x\}$ and a set of (possibly non-uniform) states $\{\ket{\psi_n}\}_{n\geq 1}$ of $d(n)$ qubits for some polynomial $d$, such that for all strings $x\in\{0,1\}^n$:
\begin{enumerate}

\item If $x\in\mathcal{L}$ then $\mathcal{Q}_x$ accepts with probability at least $2/3$ given $\ket{x}\ket{\psi_n}$ as input.

\item If $x\notin\mathcal{L}$ then $\mathcal{Q}_x$ accepts with probability at most $1/3$ given $\ket{x}\ket{\psi_n}$.

\end{enumerate}
\end{definition}

The class of decision problems for which a ``yes'' outcome can be verified in quantum poly-time, with help from a poly-size quantum proof, or witness, state, denoted $\bold{QMA}$, is defined as follows. 

\begin{definition}
$\bold{QMA}$ is the set of languages $\mathcal{L}\subseteq\{0,1\}^n$ for which there exists a poly-time uniform quantum circuit $\{\mathcal{Q}_x\}$ and a polynomial $d$, such that for all strings $x\in\{0,1\}^n$:
\begin{enumerate}

\item If $x\in\mathcal{L}$ then there exists a $d(n)$-qubit quantum proof $\ket{\phi}$ such that $\mathcal{Q}_x$ accepts with probability at least $2/3$ given $\ket{x}\ket{\phi}$ as input.

\item If $x\notin\mathcal{L}$ then $\mathcal{Q}_x$ accepts with probability at most $1/3$ given $\ket{x}\ket{\phi}$ as input, for all proofs $\ket{\phi}$.

\end{enumerate}
\end{definition}

The existential quantifiers in the above definition of $\bold{QMA}$ rigorously capture the notion of a quantum circuit having to `verify' the quantum proof. 
\end{appendices}

\end{document}